\documentclass[11pt]{article}

\usepackage[a4paper,margin=2.4cm]{geometry}
\usepackage{amsmath,amssymb,amsthm,bm}
\usepackage{booktabs}
\usepackage{graphicx}
\usepackage{xcolor}
\usepackage{hyperref}
\usepackage{microtype}
\usepackage{float}
\usepackage{placeins}

\hypersetup{
  colorlinks=true,
  citecolor=blue,
  linkcolor=blue,
  urlcolor=blue
}

\newcommand{\Lag}{\mathcal{L}}

\newcommand{\dd}{\mathrm{d}}
\newcommand{\order}{\mathcal{O}}

\newtheorem{proposition}{Proposition}

\title{\textbf{Colored Black Holes in Logarithmic Nonlinear
Yang--Mills Theory}}

\author{%
Celio R. Muniz\\[2mm]
\small Faculdade de Educa\c{c}\~ao, Ci\^encias e Letras de Iguatu,\\
\small Universidade Estadual do Cear\'a (UECE),\\
\small Iguatu, Cear\'a 63500-000, Brazil\\[1mm]
\small \texttt{celio.muniz@uece.br}
}

\date{}

\begin{document}

\maketitle

\begin{abstract}
We construct static, spherically symmetric, and asymptotically flat colored black holes in four-dimensional Einstein gravity coupled to a logarithmic nonlinear Yang--Mills field with gauge group SU(2). Unlike solutions based on the Wu--Yang ansatz, the gauge sector contains a dynamical radial amplitude, and the configurations arise from a genuinely coupled nonlinear boundary-value problem. An integral identity excludes nontrivial sign-definite solutions approaching a magnetically neutral Yang--Mills vacuum, implying that colored configurations must be nodal. We numerically construct the fundamental one-node branch and verify its convergence to the ordinary Einstein--Yang--Mills colored black hole in the linear limit. Increasing the logarithmic nonlinearity lowers the ADM mass, raises the Hawking temperature, and displaces the gauge-field node toward larger radii. Branch termination is not detected over the full range of parameters investigated. Instead, the exterior geometry approaches Schwarzschild on fixed radial domains, while the colored structure develops an increasingly extended tail. Although no independent asymptotic Yang--Mills charge is present, the non-Abelian hair leaves a subleading imprint on the far-field geometry. Inside the event horizon, the representative nonlinear solutions possess no Cauchy horizon and approach a Schwarzschild-type spacelike curvature singularity, with a finite limiting mass and no oscillary mass inflation. The recovery of the linear theory is therefore nonuniform near the cetonter: solutions that are nearly linear in the exterior eventually enter the strongly logarithmic regime in the deep interior. Finally, turning points of the Hawking temperature produce divergences and sign changes in heat capacity, indicating transitions of local thermodynamic stability.
\end{abstract}

\section{Introduction}
\label{sec:introduction}

The uniqueness theorems of electrovacuum general relativity led to the
expectation that stationary black holes should be characterized by only a
small set of asymptotic charges.  Matter fields with non-Abelian
self-interactions revealed, however, that this expectation is not universal
\cite{Bekenstein:1996pn,Volkov:1998cc}.  The discovery of the globally regular
Bartnik--McKinnon solutions \cite{Bartnik:1988am}, followed by the construction
of their black-hole counterparts \cite{Bizon:1990sr,Volkov:1989fi}, established
that four-dimensional Einstein--Yang--Mills (EYM) theory admits static,
spherically symmetric, asymptotically flat configurations with a nontrivial
radial gauge amplitude.  Rigorous existence results
\cite{Smoller:1991zy,Smoller:1993yq} and extensive numerical studies subsequently
clarified that both the solitons and the black holes form discrete families
labelled by the number of nodes of the Yang--Mills function
\cite{Kunzle:1990is,Volkov:1998cc}.

These configurations are the paradigmatic examples of ``colored'' objects.
Their Yang--Mills field approaches a vacuum at spatial infinity and carries no
independent global non-Abelian charge, although it remains nontrivial in the
bulk.  The corresponding hair is therefore secondary and cannot be inferred
from a Gauss-law charge measured at infinity
\cite{Bizon:1990sr,Volkov:1998cc}.  This evasion of the simplest no-hair
expectations comes at a price: the asymptotically flat EYM solitons and colored
black holes possess sphaleronic and/or gravitational unstable modes
\cite{Straumann:1990as,Straumann:1990wn,Lavrelashvili:1994rp}.
Nevertheless, they remain fundamental laboratories for non-Abelian gravity.
Moreover, changing the asymptotic structure can alter the stability picture:
stable Yang--Mills hair exists for suitable asymptotically anti--de Sitter
solutions \cite{Winstanley:1998sn,Bjoraker:2000qd}.  The colored solutions have
also motivated investigations of isolated-horizon mass relations
\cite{Corichi:1999nw} and of the nontrivial interior dynamics of non-Abelian
black holes \cite{Donets:1996sv,Galtsov:1997kp}.

Nonlinear gauge theories provide a natural setting in which to ask how this
structure changes away from the quadratic Yang--Mills regime.  The prototype
is Born--Infeld electrodynamics \cite{Born:1934gh}, originally introduced to
soften the short-distance behavior of the electromagnetic field.  Its
non-Abelian generalizations support finite-energy flat-space configurations,
gravitating solitons, and black holes
\cite{Galtsov:1999ft,Dyadichev:2000dy,Wirschins:2000tx}.  In particular,
non-Abelian Einstein--Born--Infeld solutions continuously deform the EYM
families and display a nontrivial dependence on the nonlinear scale.
Power-law Yang--Mills sources and related nonlinear models have likewise been
used to construct black holes in Einstein, Gauss--Bonnet, Lovelock, and
modified-gravity theories
\cite{Mazharimousavi:2008ap,Mazharimousavi:2009ti,Mazharimousavi:2011kp}.
These developments show that nonlinear gauge dynamics can modify the domain of
existence, horizon structure, mass, and thermodynamics of Yang--Mills black
holes without changing the gravitational action.

The logarithmic model considered here belongs to the same broad class of
strong-field deformations.  Its Lagrangian,
\(\mathcal{L}(X)=-\beta^2\ln(1+X/\beta^2)\), reduces to ordinary Yang--Mills
theory when \(X/\beta^2\ll1\) and introduces a single nonlinear scale
\(\beta\).  Logarithmic gauge Lagrangians were first studied in the Abelian context as alternatives to Born--Infeld electrodynamics
\cite{Soleng:1995kn}.  More recently, analytic black-hole solutions sourced by Born--Infeld, exponential, and logarithmic nonlinear Yang--Mills fields have been obtained in several dimensions \cite{Jahromi:2023logym}. Those analytic constructions employ a generalized Wu--Yang magnetic ansatz,
as do many exact black holes with standard or nonlinear Yang--Mills sources
\cite{Yasskin:1975ag,Mazharimousavi:2008ap,Mazharimousavi:2009ti,
Mazharimousavi:2011kp}.

The distinction between that sector and the present one is essential.  For a
Wu--Yang configuration the magnetic invariant is an algebraic function of the
radial coordinate, and the matter equations reduce sufficiently for the
metric to be obtained by quadrature.  The configuration is characterized by
an embedded magnetic charge and contains no dynamical radial gauge amplitude.
By contrast, the spherically symmetric \(SU(2)\) ansatz used for colored EYM
black holes contains a function \(w(r)\).  Both \(w'(r)\) and
\(1-w^2(r)\) contribute to the gauge invariant, so the Einstein and nonlinear
Yang--Mills equations form a genuinely coupled boundary-value problem.  The
resulting solutions are not analytic continuations of Wu--Yang black holes:
they have a nodal gauge profile, approach \(w=\pm1\) at infinity, and carry
colored hair rather than an independently specifiable magnetic charge.

In this work we construct static, spherically symmetric, asymptotically flat
colored black holes in four-dimensional Einstein gravity coupled to
logarithmic nonlinear \(SU(2)\) Yang--Mills theory.  We derive the reduced
field equations and the regular nonextremal-horizon data, and obtain the
far-field expansion to sufficiently high order to expose the gravitational
imprint of the hair.  In particular, the metric is Schwarzschild only at
leading order; the first colored correction appears at order \(r^{-4}\), while
the explicitly logarithmic contribution is more strongly suppressed.  We
also prove an integral identity excluding nontrivial sign-definite solutions
that remain within a single vacuum sector when \(P=-\mathcal{L}_X>0\).

We then solve the nonlinear shooting problem and construct the fundamental
one-node branch.  Varying the dimensionless logarithmic parameter and the
horizon radius, we determine the gauge and metric profiles, ADM mass, Hawking
temperature, node position, and thermodynamic response.  The solutions
approach the ordinary EYM colored black holes smoothly as
\(\beta\rightarrow\infty\), whereas stronger logarithmic nonlinearities
produce sizeable changes in their mass, temperature, and radial distribution
of the hair.  A diagnostic continuation inside the event horizon is used to
test for additional zeros of the metric function and to characterize the
termination of the numerical evolution; this analysis supports a single
resolved horizon for the representative solutions, without asserting global
regularity at the central singularity.  Finally, temperature turning points
are used to identify divergences and sign changes of the fixed-\(\beta\) heat
capacity.  These are interpreted as changes of local thermodynamic stability,
not as evidence for a global canonical phase transition.

The paper is organized as follows.  Section~\ref{sec:model} introduces the
logarithmic EYM model.  Section~\ref{sec:ansatz} specifies the spherical
non-Abelian sector and derives the dimensionless field equations.
Section~\ref{sec:boundary} develops the horizon and asymptotic expansions, and
Sec.~\ref{sec:nodal} establishes the nodal restriction.  The numerical method
is described in Sec.~\ref{sec:numerics}, while the one-node families and their
interior diagnostic are presented in Secs.~\ref{sec:results} and the following
section.  The thermodynamic properties are discussed in
Sec.~\ref{sec:thermo}, and the final section summarizes the results and open
questions.

\section{Logarithmic Einstein--Yang--Mills theory}
\label{sec:model}

We consider four-dimensional Einstein gravity minimally coupled to a
logarithmic nonlinear \(SU(2)\) Yang--Mills field. The action is
\begin{equation}
 I=\int \dd^4x\,\sqrt{-g}
 \left[
 \frac{R}{16\pi G}+\Lag(X)
 \right],
 \label{eq:action}
\end{equation}
where \(G\) is Newton's constant and
\begin{equation}
 \Lag(X)
 =-\beta^2\ln\left(1+\frac{X}{\beta^2}\right),
 \qquad
 X=\frac14 F_{\mu\nu}^{a}F^{a\mu\nu}.
 \label{eq:loglag}
\end{equation}
The parameter \(\beta\) sets the nonlinear gauge-field scale, while
\(a=1,2,3\) denotes an adjoint \(SU(2)\) index. Reality of the
Lagrangian requires
\begin{equation}
 1+\frac{X}{\beta^2}>0.
 \label{eq:realitycondition}
\end{equation}
This condition is automatically satisfied in the purely magnetic
black-hole exterior considered below, where \(N\geq0\) and hence
\(X\geq0\).

The non-Abelian field strength is defined by
\begin{equation}
 F_{\mu\nu}^{a}
 =\partial_\mu A_\nu^{a}-\partial_\nu A_\mu^{a}
 +g\,\epsilon^{abc}A_\mu^{b}A_\nu^{c},
 \label{eq:fieldstrength}
\end{equation}
where \(g\) is the Yang--Mills coupling and \(\epsilon^{abc}\) are the
\(SU(2)\) structure constants. It is convenient to introduce the
positive function
\begin{equation}
 P(X)\equiv-\Lag_X
 =\frac{1}{1+X/\beta^2}.
 \label{eq:Pdef}
\end{equation}
In the black-hole exterior, \(0<P\leq1\), with \(P\rightarrow1\) in
the linear Yang--Mills limit. In the interior, \(X\) need not remain
nonnegative, although the reality condition \(1+X/\beta^2>0\) must
still hold.

Variation of the action with respect to the metric gives
\begin{equation}
 G_{\mu\nu}=8\pi G\,T_{\mu\nu},
 \label{eq:einstein}
\end{equation}
where
\begin{equation}
 T_{\mu\nu}
 =
 g_{\mu\nu}\Lag
 +P F_{\mu\lambda}^{a}F_{\nu}^{a\ \lambda}.
 \label{eq:energymomentum}
\end{equation}
Variation with respect to the gauge potential yields the nonlinear
Yang--Mills equation
\begin{equation}
 D_\mu\left(PF^{a\mu\nu}\right)=0,
 \label{eq:gaugegeneral}
\end{equation}
where \(D_\mu\) denotes the gauge-covariant derivative in the adjoint
representation.

Unlike ordinary Yang--Mills theory, the logarithmic model possesses a
nonvanishing stress-energy trace,
\begin{equation}
 T\equiv T^\mu{}_\mu
 =4\left[\Lag(X)+X P(X)\right].
 \label{eq:trace}
\end{equation}
Thus, the nonlinear scale \(\beta\) explicitly breaks the classical
scale invariance of the Yang--Mills sector. The trace vanishes
continuously as the linear theory is recovered.

Indeed, for \(X/\beta^2\ll1\), the logarithmic Lagrangian admits the
expansion
\begin{equation}
 \Lag(X)
 =
 -X+\frac{X^2}{2\beta^2}
 -\frac{X^3}{3\beta^4}
 +\order\left(\frac{X^4}{\beta^6}\right).
 \label{eq:weakexpansion}
\end{equation}
Consequently,
\begin{equation}
 \beta\rightarrow\infty
 \quad\Longrightarrow\quad
 \Lag(X)\rightarrow-X,
 \qquad
 P(X)\rightarrow1,
\end{equation}
and the ordinary Einstein--Yang--Mills theory is recovered.

\section{Spherically symmetric non-Abelian sector}
\label{sec:ansatz}

We consider a static and spherically symmetric spacetime described by
\begin{equation}
 \dd s^2
 =
 -N(r)S^2(r)\dd t^2
 +\frac{\dd r^2}{N(r)}
 +r^2\left(\dd\theta^2+\sin^2\theta\,\dd\phi^2\right),
 \label{eq:metric}
\end{equation}
where
\begin{equation}
 N(r)=1-\frac{2Gm(r)}{r}.
 \label{eq:massfunction}
\end{equation}
Here \(m(r)\) is the dimensional mass function, while \(S(r)\) is the
redshift function. The normalization of the time coordinate will be
fixed by imposing \(S(\infty)=1\).

For the purely magnetic \(SU(2)\) connection, we adopt the standard
spherically symmetric ansatz
\begin{equation}
 A=\frac{1-w(r)}{g}
 \left(
 \tau_\phi\,\dd\theta
 -\tau_\theta\sin\theta\,\dd\phi
 \right),
 \label{eq:gaugeansatz}
\end{equation}
where \(\tau_r\), \(\tau_\theta\), and \(\tau_\phi\) form the usual
angle-dependent spherical basis in the \(SU(2)\) algebra. The
corresponding nonvanishing field-strength components are, up to the
chosen generator convention,
\begin{align}
 F_{r\theta}&=-\frac{w'}{g}\tau_\phi,
 \\
 F_{r\phi}&=\frac{w'}{g}\tau_\theta\sin\theta,
 \\
 F_{\theta\phi}&=\frac{1-w^2}{g}\tau_r\sin\theta.
 \label{eq:Fcomponents}
\end{align}
The radial gauge amplitude \(w(r)\) is a dynamical function determined
by the nonlinear Yang--Mills equation. The configurations considered
here are therefore genuinely colored and are not restricted to the
Wu--Yang sector. The constant values \(w=\pm1\) correspond to
pure-gauge vacuum configurations, whereas \(w=0\) gives an embedded
Abelian magnetic configuration.

Substitution of Eq.~\eqref{eq:gaugeansatz} into the gauge invariant
defined in Eq.~\eqref{eq:loglag} gives
\begin{equation}
 X_{\rm phys}
 =
 \frac{Nw'^2}{g^2r^2}
 +\frac{(1-w^2)^2}{2g^2r^4}.
 \label{eq:Xdimensional}
\end{equation}
The first term originates from the radial-angular components of the
field strength, while the second is associated with the purely angular
magnetic component.

To write the reduced equations in dimensionless form, we introduce the
natural EYM length scale
\begin{equation}
 \ell_{\rm EYM}=\frac{\sqrt{4\pi G}}{g},
 \label{eq:eymscale}
\end{equation}
together with
\begin{equation}
 x=\frac{r}{\ell_{\rm EYM}},
 \qquad
 \mu(x)=\frac{Gm(r)}{\ell_{\rm EYM}},
 \qquad
 b=g\beta\ell_{\rm EYM}^{\,2}.
 \label{eq:dimensionless}
\end{equation}
The dimensionless gauge invariant is
\begin{equation}
 \mathcal{X}
 \equiv g^2\ell_{\rm EYM}^{\,4}X_{\rm phys}.
 \label{eq:Xrescaling}
\end{equation}
In what follows, we drop the symbol \(x\), relabel \(\mu\) as \(m\),
and denote the dimensionless invariant \(\mathcal{X}\) again by \(X\).
The relevant quantities then become
\begin{equation}
 N=1-\frac{2m}{r},
 \qquad
 X=
 \frac{Nw'^2}{r^2}
 +\frac{(1-w^2)^2}{2r^4},
 \qquad
 P=\frac{1}{1+X/b^2},
 \label{eq:dimensionlessXP}
\end{equation}
where all radial derivatives from this point onward are taken with
respect to the dimensionless coordinate \(r\). The parameter \(b\)
measures the strength of the logarithmic corrections: finite and
smaller values of \(b\) correspond to stronger nonlinear effects,
whereas \(b\to\infty\) gives the linear EYM limit.

The Einstein and nonlinear Yang--Mills equations reduce to the coupled
ordinary differential system
\begin{align}
 m'
 &=
 r^2b^2\ln\left(1+\frac{X}{b^2}\right),
 \label{eq:mprime}
 \\
 \frac{S'}{S}
 &=
 \frac{2Pw'^2}{r},
 \label{eq:Sprime}
 \\
 \left(SNPw'\right)'
 &=
 \frac{SP}{r^2}w(w^2-1).
 \label{eq:wprime}
\end{align}
Equations~\eqref{eq:mprime}--\eqref{eq:wprime}, supplemented by regular
horizon data and asymptotically flat boundary conditions, determine
the three functions \(m(r)\), \(S(r)\), and \(w(r)\).

As a consistency check, when \(b\rightarrow\infty\),
\begin{equation}
 P\rightarrow1,
 \qquad
 b^2\ln\left(1+\frac{X}{b^2}\right)\rightarrow X,
\end{equation}
and the reduced equations become
\begin{align}
 m'
 &=
 Nw'^2+\frac{(1-w^2)^2}{2r^2},
 \label{eq:mprimeEYM}
 \\
 \frac{S'}{S}
 &=
 \frac{2w'^2}{r},
 \label{eq:SprimeEYM}
 \\
 (SNw')'
 &=
 \frac{S}{r^2}w(w^2-1),
 \label{eq:wprimeEYM}
\end{align}
which are precisely the standard magnetic EYM equations.

\section{Boundary conditions and local solutions}
\label{sec:boundary}

The reduced equations form a nonlinear boundary-value problem.
A colored black-hole solution must be analytic at a nonextremal event
horizon and approach a magnetically neutral Yang--Mills vacuum at
spatial infinity. We derive below the corresponding local expansions,
which provide both the initial data for the numerical integration and
the asymptotic identification of the non-Abelian hair.

\subsection{Near-horizon expansion}
\label{sec:near-horizon}

Let \(r=r_h\) be a nonextremal event horizon satisfying
\begin{equation}
 N(r_h)=0,
 \qquad
 N'_h>0.
 \label{eq:horizon_conditions}
\end{equation}
Since \(N=1-2m/r\) in dimensionless variables, the first condition
implies
\begin{equation}
 m(r_h)=\frac{r_h}{2}.
 \label{eq:horizon_mass}
\end{equation}
Writing \(\Delta=r-r_h\), we introduce the local expansions
\begin{align}
 m(r)
 &=
 \frac{r_h}{2}+m_1\Delta+m_2\Delta^2
 +\order(\Delta^3),
 \label{eq:horizon_expansion_m}
 \\
 N(r)
 &=
 N_1\Delta+N_2\Delta^2+\order(\Delta^3),
 \label{eq:horizon_expansion_N}
 \\
 w(r)
 &=
 w_h+w_1\Delta+w_2\Delta^2+\order(\Delta^3),
 \label{eq:horizon_expansion_w}
 \\
 S(r)
 &=
 S_h\left[
 1+s_1\Delta+s_2\Delta^2+\order(\Delta^3)
 \right].
 \label{eq:horizon_expansion_S}
\end{align}

At the horizon, the derivative contribution to the Yang--Mills
invariant vanishes because \(N(r_h)=0\). Hence,
\begin{equation}
 X_h=\frac{(1-w_h^2)^2}{2r_h^4},
 \qquad
 P_h=\frac{1}{1+X_h/b^2}.
 \label{eq:XP_horizon}
\end{equation}
The leading coefficients follow directly from the reduced field
equations:
\begin{align}
 m_1
 &=
 r_h^2b^2
 \ln\left[
 1+\frac{(1-w_h^2)^2}{2b^2r_h^4}
 \right],
 \label{eq:m1_horizon}
 \\
 N_1
 &=
 \frac{1-2m_1}{r_h},
 \label{eq:N1_horizon}
 \\
 w_1
 &=
 \frac{w_h(w_h^2-1)}{r_h^2N_1},
 \label{eq:w1_horizon}
 \\
 s_1
 &=
 \frac{2P_hw_1^2}{r_h}.
 \label{eq:s1_horizon}
\end{align}
Notice that \(P_h\) cancels from the leading-order Yang--Mills
equation, although the logarithmic nonlinearity still affects \(w_1\)
indirectly through \(N_1\).

The nonextremality condition \(N_1>0\) becomes
\begin{equation}
 1-2r_h^2b^2
 \ln\left[
 1+\frac{(1-w_h^2)^2}{2b^2r_h^4}
 \right]>0.
 \label{eq:nonextremality}
\end{equation}
For fixed \(r_h\) and \(b\), this inequality restricts the admissible
values of the shooting parameter \(w_h\). The extremal case
\(N_1=0\) is not described by the present expansion and would require
a separate near-horizon analysis.

For completeness, the second-order coefficients can also be obtained
algebraically. Defining
\begin{equation}
 X_1=
 \frac{N_1w_1^2}{r_h^2}
 -\frac{2w_hw_1(1-w_h^2)}{r_h^4}
 -\frac{2(1-w_h^2)^2}{r_h^5},
 \label{eq:X1_horizon}
\end{equation}
we find
\begin{equation}
 P_1=-\frac{P_h^2}{b^2}X_1,
 \qquad
 m_2=\frac{m_1}{r_h}+\frac{r_h^2}{2}P_hX_1,
 \qquad
 N_2=-\frac{N_1}{r_h}-\frac{2m_2}{r_h}.
 \label{eq:secondorder_metric}
\end{equation}
Introducing
\begin{equation}
 Q_h=w_h(w_h^2-1),
 \qquad
 Q'_h=3w_h^2-1,
 \label{eq:horizon_Q}
\end{equation}
and
\begin{equation}
 \mathcal{R}_1=
 \frac{(s_1P_h+P_1)Q_h}{r_h^2}
 +\frac{P_hQ'_hw_1}{r_h^2}
 -\frac{2P_hQ_h}{r_h^3},
 \label{eq:R1_horizon}
\end{equation}
the second-order gauge coefficient is
\begin{equation}
 w_2=
 \frac{
 \mathcal{R}_1
 -2P_hN_2w_1
 -2N_1P_1w_1
 -2s_1N_1P_hw_1
 }
 {4N_1P_h}.
 \label{eq:w2_horizon}
\end{equation}
Finally, defining
\begin{equation}
 a_1=
 2\left[
 \frac{P_1w_1^2}{r_h}
 +\frac{4P_hw_1w_2}{r_h}
 -\frac{P_hw_1^2}{r_h^2}
 \right],
 \label{eq:a1_horizon}
\end{equation}
we obtain
\begin{equation}
 s_2=\frac{a_1+s_1^2}{2}.
 \label{eq:s2_horizon}
\end{equation}
The first-order coefficients are sufficient to initialize the outward
integration, whereas the second-order expansion provides an independent
check of horizon regularity and numerical convergence.

After fixing the asymptotic normalization \(S(\infty)=1\), the Hawking
temperature is
\begin{equation}
 T_H=\frac{S_hN_1}{4\pi\ell_{\rm EYM}}.
 \label{eq:Hawking_temperature_horizon}
\end{equation}
In the numerical integration, it is convenient to set the unnormalized
horizon value to \(S_h=1\). If the resulting asymptotic value is
\(S_\infty\), the physical normalization is obtained by
\(S\rightarrow S/S_\infty\), and therefore
\begin{equation}
 \tau\equiv\ell_{\rm EYM}T_H
 =\frac{N_1}{4\pi S_\infty}.
 \label{eq:Hawking_temperature_normalized}
\end{equation}

\subsection{Asymptotic expansion and colored-hair imprint}
\label{sec:asymptotic-expansion}

Asymptotic flatness and the normalization of the time coordinate require
\begin{equation}
 N(r)\rightarrow1,
 \qquad
 S(r)\rightarrow1,
 \qquad
 m(r)\rightarrow M,
 \label{eq:asymptotic_flatness}
\end{equation}
where \(M\) is the dimensionless ADM mass. The Yang--Mills equation
allows the constant asymptotic values
\begin{equation}
 w(\infty)=0,\ \pm1.
\end{equation}
The choice \(w(\infty)=0\) describes an embedded Abelian configuration
with nonzero magnetic charge. In contrast, colored black holes are
magnetically neutral and approach one of the pure-gauge vacua,
\begin{equation}
 w(\infty)=\varepsilon,
 \qquad
 \varepsilon=\pm1.
 \label{eq:YM_vacua}
\end{equation}
For the one-node branch studied below, we choose \(w_h>0\) and
\(w(\infty)=-1\).

We therefore write
\begin{equation}
 w(r)=
 -1+\frac{c}{r}
 +\frac{d}{r^2}
 +\frac{e}{r^3}
 +\order(r^{-4}),
 \label{eq:w_asymptotic_general}
\end{equation}
where \(c\) controls the leading asymptotic colored-field amplitude.
Substitution into the field equations determines
\begin{equation}
 d=\frac{3}{2}Mc-\frac{3}{4}c^2,
 \qquad
 e=
 \frac{12}{5}M^2c
 -\frac{21}{10}Mc^2
 +\frac{11}{20}c^3.
 \label{eq:de_asymptotic}
\end{equation}
Thus,
\begin{equation}
 w(r)=
 -1+\frac{c}{r}
 +\frac{\frac32Mc-\frac34c^2}{r^2}
 +\frac{
 \frac{12}{5}M^2c
 -\frac{21}{10}Mc^2
 +\frac{11}{20}c^3
 }{r^3}
 +\order(r^{-4}).
 \label{eq:w_asymptotic}
\end{equation}

The corresponding gravitational functions are
\begin{align}
 m(r)
 &=
 M-\frac{c^2}{r^3}
 +\frac{2c^3-\frac52Mc^2}{r^4}
 +\order(r^{-5}),
 \label{eq:m_asymptotic}
 \\
 N(r)
 &=
 1-\frac{2M}{r}
 +\frac{2c^2}{r^4}
 +\frac{5Mc^2-4c^3}{r^5}
 +\order(r^{-6}),
 \label{eq:N_asymptotic}
 \\
 S(r)
 &=
 1-\frac{c^2}{2r^4}
 +\frac{6c^3-12Mc^2}{5r^5}
 +\order(r^{-6}).
 \label{eq:S_asymptotic}
\end{align}
Consequently,
\begin{equation}
 -g_{tt}=N(r)S^2(r)
 =
 1-\frac{2M}{r}
 +\frac{c^2}{r^4}
 +\frac{11Mc^2-8c^3}{5r^5}
 +\order(r^{-6}).
 \label{eq:gtt_asymptotic}
\end{equation}
The spacetime is therefore Schwarzschild only at leading order. The
first gravitational imprint of the colored hair appears at order
\(r^{-4}\), rather than through a Reissner--Nordstr\"om-type
\(r^{-2}\) term.

The coefficient \(c\) is not an independently conserved charge. Once
\(r_h\), \(b\), the node number, and the asymptotic vacuum are fixed,
\(c\) is determined by the global shooting problem. It therefore
characterizes secondary non-Abelian hair.

The Yang--Mills invariant behaves asymptotically as
\begin{equation}
 X(r)=\frac{3c^2}{r^6}+\order(r^{-7}),
 \label{eq:X_asymptotic}
\end{equation}
so that
\begin{equation}
 P(r)
 =
 1-\frac{3c^2}{b^2r^6}
 +\order(r^{-7}).
 \label{eq:P_asymptotic}
\end{equation}
Moreover,
\begin{equation}
 b^2\ln\left(1+\frac{X}{b^2}\right)
 =
 X-\frac{X^2}{2b^2}
 +\order\left(\frac{X^3}{b^4}\right),
 \label{eq:L_asymptotic}
\end{equation}
and hence the leading explicitly logarithmic correction decays as
\begin{equation}
 \frac{X^2}{2b^2}
 =
 \frac{9c^4}{2b^2r^{12}}
 +\order(r^{-13}).
 \label{eq:logarithmic_asymptotic_correction}
\end{equation}
After radial integration in the mass equation, this produces metric
corrections beginning schematically at
\begin{equation}
 \delta N_{\rm log}
 =
 \order\left(\frac{1}{b^2r^{10}}\right).
 \label{eq:log_metric_order}
\end{equation}
The logarithmic interaction can therefore modify the near-horizon and
intermediate regions substantially while becoming strongly suppressed
in the far field, where the solution rapidly approaches its ordinary
EYM asymptotics.

No regular expansion about \(r=0\) is required for the black-hole
solutions, whose exterior domain begins at \(r=r_h>0\). A regular
origin would instead define a distinct particle-like solitonic problem,
analogous to the Bartnik--McKinnon solutions, which is not considered
here.

\section{Analytic restrictions on nodeless configurations}
\label{sec:nodal}

Before solving the boundary-value problem numerically, one can derive
a useful restriction directly from the Yang--Mills equation. The
argument relies only on asymptotic flatness and on the positivity of
\(S\), \(N\), and \(P\) throughout the black-hole exterior. It is
therefore not specific to the detailed logarithmic form of the
Lagrangian.

\begin{proposition}
\label{prop:nodeless}
Let \(w(r)\) be an asymptotically flat, magnetically neutral solution
with
\begin{equation}
 w(\infty)=\varepsilon,
 \qquad
 \varepsilon=\pm1.
\end{equation}
If
\begin{equation}
 \varepsilon w(r)>0
\end{equation}
everywhere outside a nonextremal event horizon, then
\begin{equation}
 w(r)\equiv\varepsilon.
\end{equation}
\end{proposition}

\begin{proof}
Consider first the sector \(w(\infty)=1\). Multiplying
Eq.~\eqref{eq:wprime} by \(w-1\) and integrating from the horizon to
spatial infinity gives
\begin{align}
 \left[
 (w-1)SNPw'
 \right]_{r_h}^{\infty}
 -
 \int_{r_h}^{\infty}SNPw'^2\,\dd r
 &=
 \int_{r_h}^{\infty}
 \frac{SP}{r^2}(w-1)w(w^2-1)\,\dd r.
\end{align}
The boundary term vanishes. At the horizon this follows from
\(N(r_h)=0\) and the regularity of \(w\) and \(w'\), while at infinity
it follows from \(w-1=\order(r^{-1})\) and
\(w'=\order(r^{-2})\). Consequently,
\begin{equation}
 -\int_{r_h}^{\infty}SNPw'^2\,\dd r
 =
 \int_{r_h}^{\infty}
 \frac{SP}{r^2}w(w-1)^2(w+1)\,\dd r.
 \label{eq:nodalidentity_positive}
\end{equation}
If \(w>0\) throughout the exterior, the left-hand side is
nonpositive, whereas the right-hand side is nonnegative. Equality can
therefore hold only if
\begin{equation}
 w'\equiv0.
\end{equation}
The asymptotic condition then implies \(w\equiv1\).

For \(w(\infty)=-1\), we instead multiply
Eq.~\eqref{eq:wprime} by \(w+1\). The same integration by parts gives
\begin{equation}
 -\int_{r_h}^{\infty}SNPw'^2\,\dd r
 =
 \int_{r_h}^{\infty}
 \frac{SP}{r^2}w(w+1)^2(w-1)\,\dd r.
 \label{eq:nodalidentity_negative}
\end{equation}
When \(w<0\), the right-hand side is again nonnegative, and hence
\(w'\equiv0\). The boundary condition now yields \(w\equiv-1\).
\end{proof}

An immediate consequence is that every nontrivial, magnetically
neutral colored black hole must possess at least one zero of \(w(r)\).
Indeed, a solution approaching \(w(\infty)=\varepsilon\) has the sign
of \(\varepsilon\) sufficiently far from the black hole. If it
preserved that sign throughout the exterior, Proposition
\ref{prop:nodeless} would force it to be the corresponding vacuum
solution. A nontrivial configuration must therefore cross \(w=0\).

The nodes are necessarily simple. Suppose that at some exterior point
\(r=r_0>r_h\),
\begin{equation}
 w(r_0)=0,
 \qquad
 w'(r_0)=0.
\end{equation}
Since \(N(r_0)>0\), the reduced field equations form a regular
initial-value problem at \(r_0\). Uniqueness would then imply
\(w(r)\equiv0\), corresponding to the embedded magnetic configuration,
which is incompatible with \(w(\infty)=\pm1\). Therefore, a colored
solution can only cross zero with
\begin{equation}
 w'(r_0)\neq0.
\end{equation}

The number of nodes,
\begin{equation}
 n=\#\left\{r\in(r_h,\infty):w(r)=0\right\},
 \label{eq:node_number}
\end{equation}
provides a discrete label for the colored families. Since the zeros
are simple, \(n\) remains unchanged under continuous variations of
the parameters along a regular solution branch, except possibly at a
branch endpoint or when the boundary conditions cease to be
satisfied. The solutions constructed below belong to the fundamental
colored branch with \(n=1\).

Further information follows by considering an exterior stationary
point \(r=r_\ast\), where \(w'(r_\ast)=0\). At such a point,
Eq.~\eqref{eq:wprime} reduces to
\begin{equation}
 w''(r_\ast)
 =
 \frac{w_\ast(w_\ast^2-1)}
 {N(r_\ast)r_\ast^2}.
 \label{eq:stationary_w}
\end{equation}
Since \(N(r_\ast)>0\), the sign of \(w''\) is fixed entirely by
\(w_\ast(w_\ast^2-1)\). Consequently,
\begin{align}
 0<w_\ast<1
 &\quad\Longrightarrow\quad
 w''(r_\ast)<0,
 \\
 -1<w_\ast<0
 &\quad\Longrightarrow\quad
 w''(r_\ast)>0,
 \\
 w_\ast>1
 &\quad\Longrightarrow\quad
 w''(r_\ast)>0,
 \\
 w_\ast<-1
 &\quad\Longrightarrow\quad
 w''(r_\ast)<0.
 \label{eq:stationary_restrictions}
\end{align}
Thus, within \(0<w<1\), a stationary point can only be a local
maximum, whereas within \(-1<w<0\) it can only be a local minimum.
These restrictions provide useful diagnostics for the numerical
shooting procedure and help distinguish the physical nodal branches
from overshooting solutions.

The preceding results depend on the nonlinear theory only through
\(P>0\). They consequently apply to a broader class of monotonic
nonlinear Yang--Mills models satisfying
\begin{equation}
 -\Lag_X>0
\end{equation}
in the relevant field domain. In the logarithmic model, this condition
holds automatically for the purely magnetic configurations considered
here.

\section{Numerical construction}
\label{sec:numerics}

The reduced field equations constitute a nonlinear two-point
boundary-value problem. We solve them by shooting outward from the
event horizon and adjusting the horizon value \(w_h\) until the
required Yang--Mills vacuum is reached at spatial infinity.

Unless otherwise stated, the radial profiles and representative
solutions are constructed with
\begin{equation}
 r_h=1.
 \label{eq:rhunit}
\end{equation}
For fixed \(b\), the only nontrivial shooting parameter is then
\(w_h\). The initial value of \(S_h\) is arbitrary before the
asymptotic normalization of the time coordinate and is set to unity
during the integration. The horizon expansions
\eqref{eq:horizon_expansion_m}--\eqref{eq:horizon_expansion_S},
together with Eqs.~\eqref{eq:m1_horizon}--\eqref{eq:s1_horizon},
provide regular initial data at
\begin{equation}
 r=r_h+\epsilon,
 \qquad
 \epsilon=2\times10^{-6}.
 \label{eq:numerical_epsilon}
\end{equation}
The first-order horizon expansion is sufficient at this value of
\(\epsilon\); the second-order coefficients derived in
Sec.~\ref{sec:near-horizon} were used as an independent local
consistency check.

For the fundamental colored branch, we choose
\begin{equation}
 w_h>0,
 \qquad
 w(\infty)=-1,
 \qquad
 n=1.
 \label{eq:fundamental_branch}
\end{equation}
The asymptotic expansion \eqref{eq:w_asymptotic} contains a decaying
mode proportional to \(1/r\) and an inadmissible growing mode. At a
finite outer boundary \(R\), the latter can be eliminated to leading
order by imposing the Robin condition
\begin{equation}
 \mathcal{R}_R
 \equiv
 w'(R)+\frac{w(R)+1}{R}=0.
 \label{eq:shootcondition}
\end{equation}
This condition is preferable to imposing \(w(R)=-1\) directly because
the physical solution reaches its vacuum value through the algebraic
tail
\begin{equation}
 w(r)+1=\frac{c}{r}+\order(r^{-2}).
\end{equation}

The shooting parameter is first scanned over an interval compatible
with the nonextremality condition \eqref{eq:nonextremality}. Once a
sign change of \(\mathcal{R}_R\) is identified, the corresponding root
is refined with a bracketing algorithm. Solutions that fail to reach
the outer boundary, violate exterior regularity, or possess a node
number different from the desired value are discarded. The position
of each node is subsequently obtained by root finding on the
interpolated numerical profile.

After the outward integration, the redshift function is normalized
according to
\begin{equation}
 S(r)\longrightarrow\frac{S(r)}{S(R)}.
 \label{eq:Snormalization}
\end{equation}
The dimensionless Hawking temperature is then computed from
Eq.~\eqref{eq:Hawking_temperature_normalized}, and the dimensionless
ADM mass is approximated by
\begin{equation}
 M\simeq m(R).
 \label{eq:numerical_mass}
\end{equation}
Since \(m(r)=M+\order(r^{-3})\), the omitted asymptotic mass tail is
strongly suppressed at the adopted outer boundary.

For numerical integration, the gauge equation is written as an
explicit first-order system. Defining
\begin{equation}
 y(r)=w'(r),
\end{equation}
we use
\begin{equation}
 N'=\frac{1-N}{r}-\frac{2m'}{r},
 \qquad
 \frac{S'}{S}=\frac{2Py^2}{r}.
 \label{eq:numerical_metric_derivatives}
\end{equation}
Because \(P=P(X)\) and \(X\) depends on \(y^2\), differentiation of
\(P\) introduces a term proportional to \(y'=w''\). To isolate it,
we write
\begin{equation}
 X'=X'_0+\frac{2Nyy'}{r^2},
 \label{eq:Xprime_decomposition}
\end{equation}
where
\begin{equation}
 X'_0=
 \frac{N'y^2}{r^2}
 -\frac{2Ny^2}{r^3}
 +\frac{2w(w^2-1)y}{r^4}
 -\frac{2(1-w^2)^2}{r^5}.
 \label{eq:Xprime0}
\end{equation}
Substituting
\begin{equation}
 P'=-\frac{P^2}{b^2}X'
\end{equation}
into Eq.~\eqref{eq:wprime} and collecting all terms proportional to
\(y'\) gives
\begin{equation}
 y'=
 \frac{
 \displaystyle
 \frac{Pw(w^2-1)}{r^2}
 -N'Py
 -\frac{S'}{S}NPy
 +\frac{NP^2y}{b^2}X'_0
 }
 {\displaystyle
 NP\left(
 1-\frac{2PNy^2}{b^2r^2}
 \right)}.
 \label{eq:explicit_wsecond}
\end{equation}
Thus, the terms generated by the implicit \(w''\)-dependence of \(P'\)
are retained algebraically rather than neglected or treated
perturbatively.

The radial profiles and parameter scans presented below were computed
with
\begin{equation}
 R=80,
\end{equation}
using a maximum integration step \(\Delta r=0.1\). The residual scan
was performed with relative tolerance \(10^{-8}\), while the final
profiles were recomputed with relative and absolute tolerances of
approximately \(3\times10^{-10}\) and \(3\times10^{-12}\),
respectively. For the reported solutions,
\begin{equation}
 \left|\mathcal{R}_R\right|<10^{-8}.
 \label{eq:residual_bound}
\end{equation}
The representative radial profiles and the horizon-family scans were
computed using the outer boundary \(R=80\). The numerical quantities
reported in Table~\ref{tab:numerical-solutions} were subsequently
recomputed with \(R=200\), in order to reduce the residual uncertainty
associated with truncating the asymptotic region. The convergence tests
described below were performed independently using
\(R=40,\,80,\) and \(120\).

To assess the dependence on the numerical boundary, representative
solutions with \(r_h=1\) and
\begin{equation}
 b=0.25,\ 1,\ 2,\ 100
\end{equation}
were recomputed for
\begin{equation}
 R=40,\ 80,\ 120.
\end{equation}
Increasing the outer boundary from \(R=80\) to \(R=120\) produced the
maximum absolute variations
\begin{align}
 |\Delta w_h|&<1.2\times10^{-4},
 &
 |\Delta M|&<6.8\times10^{-5},
 \\
 |\Delta r_{\rm node}|&<7.3\times10^{-5},
 &
 |\Delta\tau|&<7.3\times10^{-8}.
 \label{eq:radial_convergence}
\end{align}
The largest finite-boundary sensitivity occurs for the strongly
nonlinear case \(b=0.25\). For \(b\geq1\), the variation in the mass
between \(R=80\) and \(R=120\) remains at or below approximately
\(1.2\times10^{-5}\).

We also varied the integration tolerances from the
\(10^{-7}\)--\(10^{-8}\) level to the \(10^{-10}\)--\(10^{-11}\)
level. The resulting changes satisfy
\begin{align}
 |\Delta w_h|&<7.4\times10^{-9},
 &
 |\Delta M|&<5.7\times10^{-10},
 \\
 |\Delta r_{\rm node}|&<2.1\times10^{-8},
 &
 |\Delta\tau|&<4.7\times10^{-10}.
 \label{eq:tolerance_convergence}
\end{align}
The dominant numerical uncertainty is therefore associated with the
finite outer boundary rather than with the local integration accuracy.

In every tested case, the solution retained exactly one simple node
and satisfied
\begin{equation}
 N(r)>0,
 \qquad
 S(r)>0,
 \qquad
 r>r_h.
 \label{eq:exterior_numerical_regular}
\end{equation}

As an independent validation of the numerical implementation, we also
solved the ordinary EYM equations directly by setting \(P=1\), rather
than approximating the linear limit by a large but finite value of
\(b\). For the fundamental one-node solution with \(r_h=1\) and
\(R=200\), we obtained
\begin{equation}
 w_h^{\rm EYM}=0.63220697,
 \qquad
 M^{\rm EYM}=0.93719027,
 \label{eq:EYM_benchmark}
\end{equation}
together with
\begin{equation}
 r_{\rm node}^{\rm EYM}=2.21544085,
 \qquad
 \ell_{\rm EYM}T_H^{\rm EYM}=0.02941139.
 \label{eq:EYM_benchmark_global}
\end{equation}
The horizon value agrees with the published value
\(w_h\simeq0.6322\) for the fundamental EYM colored black hole with
\(r_h=1\)~\cite{Breitenlohner:1997ud}. Moreover, the \(b=100\)
solution reported below differs from the direct EYM benchmark by less
than \(2.1\times10^{-5}\) in relative terms for all the quantities
listed in Table~\ref{tab:numerical-solutions}. This agreement provides
a quantitative check of both the nonlinear solver and its
\(b\rightarrow\infty\) limit.

These checks support the numerical robustness of the one-node colored
families presented in the following sections.

\section{One-node colored black holes}
\label{sec:results}

We now present the fundamental one-node colored black-hole solutions.
We first fix \(r_h=1\) and examine how the logarithmic parameter \(b\)
deforms the gauge and metric profiles. We then allow the horizon radius
to vary and construct continuous families at fixed \(b\).

Representative numerical data are summarized in
Table~\ref{tab:numerical-solutions}. Each solution begins with
\(w_h>0\), crosses zero exactly once, and approaches the opposite
Yang--Mills vacuum, \(w(\infty)=-1\). Since
\(1-w^2\rightarrow0\) asymptotically, these solutions carry no
independently specifiable Yang--Mills magnetic charge, in contrast
with the logarithmic Wu--Yang sector~\cite{Jahromi:2023logym}.
The values in the table were recomputed using \(R=200\), whereas the
radial profiles and parameter scans shown below use \(R=80\).

\begin{table}[!htbp]
\centering
\caption{Representative numerical data for the fundamental one-node
colored black-hole solutions with \(r_h=1\). The entries with \(b\geq0.25\) were computed using \(R=200\),
whereas the strongly nonlinear \(b=0.1\) solution required
\(R=500\) because of its substantially longer Yang--Mills tail.}
\label{tab:numerical-solutions}
\begin{tabular}{ccccccc}
\hline\hline
\(b\) &
\(w_h\) &
\(w'_h\) &
\(M\) &
\(r_{\rm node}\) &
\(\ell_{\rm EYM}T_H\) &
\(\lvert\mathcal{R}_R\rvert\)
\\
\hline
0.10
& 0.421258
& $-0.372988$
& 0.754248
& 5.695213
& 0.072983
& $<10^{-8}$
\\
0.25
& 0.517680
& \(-0.478557\)
& 0.851673
& 3.293232
& 0.057739
& \(<10^{-8}\)
\\
0.50
& 0.594227
& \(-0.552343\)
& 0.904791
& 2.517687
& 0.041966
& \(<10^{-8}\)
\\
1
& 0.622004
& \(-0.581748\)
& 0.927804
& 2.290075
& 0.033221
& \(<10^{-8}\)
\\
2
& 0.629599
& \(-0.590369\)
& 0.934747
& 2.233835
& 0.030412
& \(<10^{-8}\)
\\
100
& 0.632206
& \(-0.593353\)
& 0.937189
& 2.215419
& 0.029412
& \(<10^{-8}\)
\\
\hline\hline
\end{tabular}
\end{table}

The dependence on \(b\) is smooth. As \(b\) increases, \(w_h\), the magnitude of \(w'_h\), and the total mass approach their EYM values, while the node moves toward the horizon and the temperature decreases
toward its EYM limit. The case \(b=100\) is already numerically close to this regime. Relative to \(b=100\), the strongly nonlinear case \(b=0.25\) gives approximately
\begin{equation}
 \frac{M_{0.25}-M_{100}}{M_{100}}
 \simeq-9.1\%,
 \qquad
 \frac{T_{0.25}-T_{100}}{T_{100}}
 \simeq+96.3\%,
 \qquad
 \frac{r_{{\rm node},0.25}-r_{{\rm node},100}}
 {r_{{\rm node},100}}
 \simeq+48.7\%.
 \label{eq:relative_log_deformation}
\end{equation}
We also continued the fundamental branch into the more strongly nonlinear regime \(b<0.25\). No lower critical value of \(b\) was found: regular one-node exterior solutions were obtained down to \(b=0.002\). As \(b\) decreases, the ADM mass approaches the Schwarzschild value \(M=r_h/2\), while the Hawking temperature tends
to \(1/(4\pi r_h)\). At the same time, the Yang--Mills node moves rapidly outwards and the asymptotic gauge tail extends over progressively larger radial scales. These results suggest that the limit \(b\rightarrow0\) is characterized not by a finite termination of the branch, but by a decoupling limit in which the exterior geometry approaches Schwarzschild while the colored structure is displaced toward infinity.

Because of this increasingly long gauge-field tail, reliable shooting at very small \(b\) requires an outer boundary much larger than that used for the representative profiles. We therefore use the
\(b<0.1\) solutions only to establish the qualitative continuation of the branch and do not quote their global quantities at the precision adopted in Table~\ref{tab:numerical-solutions}. Thus, within the parameter range explored, increasing the logarithmic nonlinearity drives the geometry toward the Schwarzschild limit while displacing the colored Yang--Mills structure to progressively larger radial scales.

The radial profiles shown in Fig.~\ref{fig:wprofiles} illustrate the initial deformation away from the EYM regime over the representative range \(0.25\leq b\leq100\). All curves interpolate between a positive
horizon value and the vacuum \(w=-1\), in agreement with the analytic nodal restriction derived in Sec.~\ref{sec:nodal}. The single zero is
simple, and its outward displacement is already pronounced for the smallest value displayed, \(b=0.25\). The solutions with \(b<0.25\) are not included in this figure because their increasingly extended
gauge-field tails require substantially larger radial domains. At large radius, each profile approaches the \(1/r\) colored-field decay derived in Sec.~\ref{sec:asymptotic-expansion}, with a solution-dependent asymptotic coefficient.
\begin{figure}[!htbp]
\centering
\includegraphics[width=0.68\linewidth]
{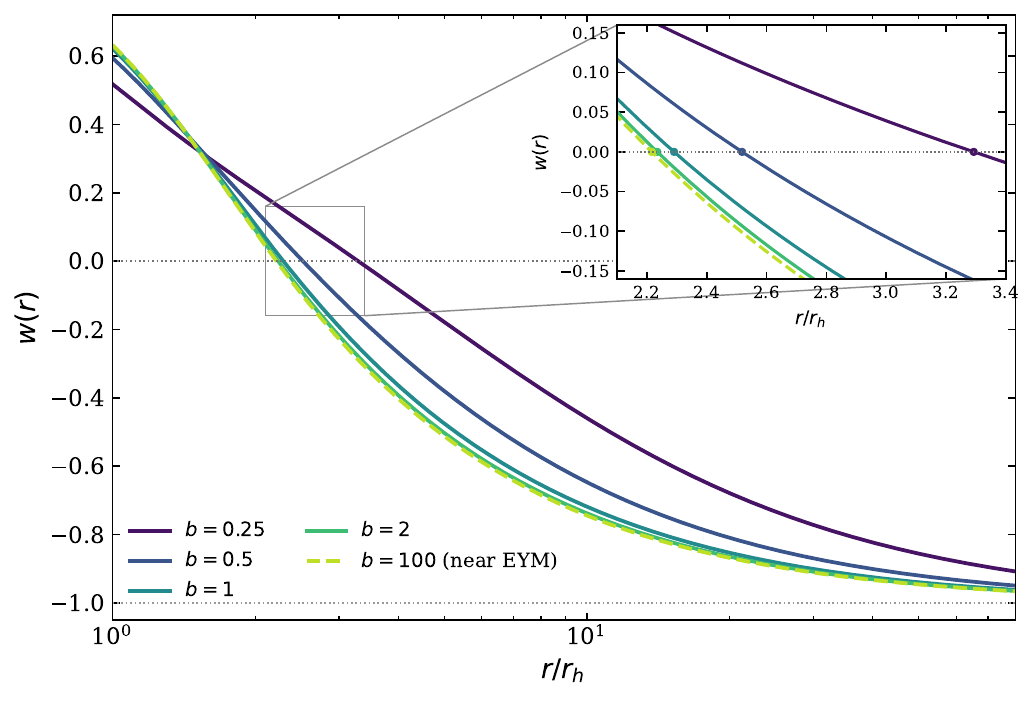}
\caption{Radial gauge-field profiles \(w(r)\) for the fundamental
one-node logarithmic Yang--Mills colored black holes with \(r_h=1\).
The inset enlarges the nodal region. Decreasing \(b\) strengthens the
logarithmic nonlinear effects, lowers \(w_h\), and shifts the node away
from the event horizon. The \(b=100\) curve represents the near-EYM
regime. All profiles approach \(w=-1\) asymptotically.}
\label{fig:wprofiles}
\end{figure}

The corresponding metric functions are displayed in
Fig.~\ref{fig:metricprofiles}. In every case, \(N(r)\) vanishes
linearly at the event horizon and remains positive throughout the
exterior. The normalized redshift function also remains positive and
approaches unity monotonically. For smaller \(b\), the reduction of
the total mass raises \(N(r)\) at fixed radius, while \(S(r)\) remains
closer to its asymptotic value. The logarithmic interaction therefore
produces its largest geometric effect in the near-horizon and
intermediate regions, consistently with its strong suppression in the
far-field expansion.

\begin{figure}[!htbp]
\centering
\includegraphics[width=\linewidth]
{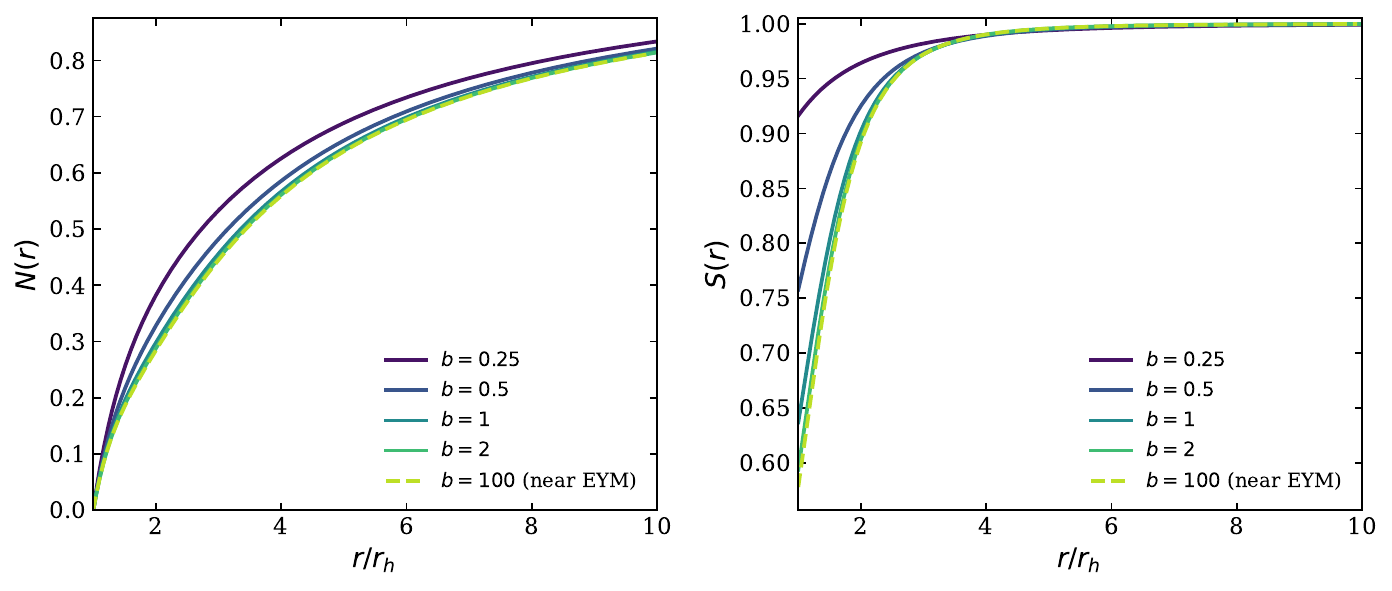}
\caption{Metric functions \(N(r)\) (left panel) and \(S(r)\)
(right panel) for the fundamental one-node colored black holes with
\(r_h=1\). The redshift function is normalized by \(S(\infty)=1\).
Strong logarithmic nonlinearities, corresponding to smaller \(b\),
reduce the asymptotic mass and weaken the variation of \(S(r)\) outside
the horizon. The profiles converge rapidly toward the near-EYM behavior
as \(b\) increases.}
\label{fig:metricprofiles}
\end{figure}

To examine the dependence on the horizon scale, we next vary \(r_h\)
while keeping \(b\) fixed. The resulting families are shown in
Fig.~\ref{fig:parameterdependence}. The ADM mass grows monotonically
with \(r_h\), whereas the normalized node position exhibits a
nonmonotonic dependence whose minimum shifts with \(b\). The Hawking
temperature displays a richer structure: for sufficiently weak
logarithmic deformation, it develops one or more turning points. These
turning points will be analyzed thermodynamically in the next section.

The separation among the curves is largest for small \(r_h\), where
the Yang--Mills invariant is strongest near the horizon. As \(r_h\)
increases, the characteristic field strength decreases relative to
the nonlinear scale \(b\), and the families approach one another.
This confirms that the logarithmic deformation primarily affects the
compact, strong-field portion of the colored branches.

\begin{figure}[!htbp]
\centering
\includegraphics[width=\linewidth]
{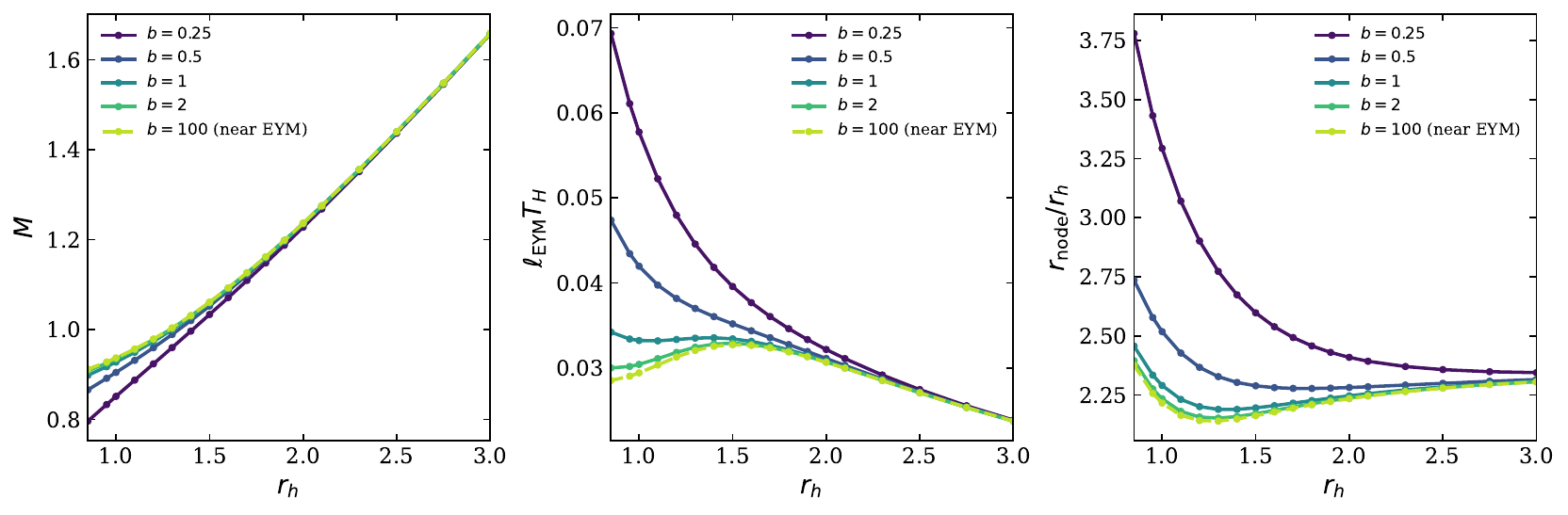}
\caption{Dimensionless ADM mass \(M\), Hawking temperature
\(\ell_{\rm EYM}T_H\), and normalized node position
\(r_{\rm node}/r_h\) as functions of the horizon radius for the
fundamental one-node colored branch. Each curve corresponds to a fixed
value of \(b\). Nonlinear effects are most pronounced for small
horizons and small \(b\). For weaker logarithmic deformation, the
temperature develops turning points, while the normalized node
position exhibits a \(b\)-dependent minimum.}
\label{fig:parameterdependence}
\end{figure}

\FloatBarrier

The present construction is restricted to the fundamental \(n=1\)
branch. In the EYM limit, colored black holes form a discrete tower
labeled by
\begin{equation}
 n=1,2,3,\ldots.
\end{equation}
The same shooting strategy can, in principle, be extended by requiring
\(n\) simple crossings and
\begin{equation}
 w(\infty)=(-1)^n
\end{equation}
for solutions beginning with \(w_h>0\). However, the existence and
global structure of the higher-node logarithmic branches have not been
established here. In particular, determining whether they possess a
critical value of \(b\), multiple coexisting solutions at fixed
\((b,r_h)\), or additional branch endpoints requires a separate
systematic survey.

\section{Internal geometry and horizon structure}
\label{sec:interior}

Having constructed the exterior one-node solutions, we now investigate
their continuation across the event horizon. In particular, we examine
whether the logarithmic Yang--Mills black holes possess additional
zeros of \(N(r)\), corresponding to inner horizons, and determine their
leading behavior near the central singularity.

A direct inward integration of the explicit second-order equation for
\(w(r)\) becomes unreliable at small radii. The radial-gradient and
angular contributions to
\begin{equation}
 X=
 \frac{Nw'^2}{r^2}
 +\frac{(1-w^2)^2}{2r^4}
\end{equation}
become individually large and nearly cancel, producing a substantial
loss of floating-point precision. We therefore reformulate the gauge
sector in terms of the flux variable
\begin{equation}
 \mathcal{H}(r)=N(r)P(r)w'(r)
 \label{eq:interior_flux}
\end{equation}
and the logarithmic argument
\begin{equation}
 \Xi(r)=1+\frac{X(r)}{b^2}
 =\frac{1}{P(r)}.
 \label{eq:Xi_definition}
\end{equation}
Equation~\eqref{eq:interior_flux} gives
\begin{equation}
 w'=\frac{\mathcal{H}\Xi}{N}.
 \label{eq:wprime_flux}
\end{equation}
Substitution into the definition of \(X\) yields an algebraic equation
for \(\Xi\),
\begin{equation}
 \mathcal{A}\Xi^2-\Xi+\mathcal{B}=0,
 \label{eq:Xi_quadratic}
\end{equation}
where
\begin{equation}\label{CursiveAandB}
 \mathcal{A}
 =
 \frac{\mathcal{H}^2}{b^2Nr^2},
 \qquad
 \mathcal{B}
 =
 1+\frac{(1-w^2)^2}{2b^2r^4}.
\end{equation}
The root continuously connected to the exterior solution is
\begin{equation}
 \Xi=
 \frac{2\mathcal{B}}
 {1+\sqrt{1-4\mathcal{A}\mathcal{B}}}.
 \label{eq:Xi_stable}
\end{equation}
This form is numerically stable and avoids evaluating \(X\) as the
difference of two nearly equal large quantities.

The interior system can consequently be written as
\begin{align}
 m'&=r^2b^2\ln\Xi,
 \label{eq:interior_mprime}\\
 \frac{S'}{S}&=\frac{2w'^2}{\Xi r},
 \label{eq:interior_Sprime}\\
 w'&=\frac{\mathcal{H}\Xi}{N},
 \label{eq:interior_wprime}\\
 \mathcal{H}'&=
 \frac{w(w^2-1)}{\Xi r^2}
 -\frac{S'}{S}\mathcal{H}.
 \label{eq:interior_Hprime}
\end{align}
For the numerical integration, we introduce
\begin{equation}
 u=-\ln\left(\frac{r}{r_h}\right),
 \label{eq:interior_coordinate}
\end{equation}
which resolves the increasingly small radial scales near the center.
The integration is initialized at \(r=r_h-\epsilon\) using the
second-order horizon expansion and the full-precision values of \(w_h\)
obtained from the exterior shooting problem.

For all representative solutions with
\begin{equation}
 b=0.25,\quad0.5,\quad1,\quad2,\quad100,
\end{equation}
the integration reaches
\begin{equation}
 \frac{r}{r_h}=10^{-8}.
\end{equation}
No additional zero of \(N(r)\) is detected, and \(\Xi(r)\) remains
positive throughout the integration domain. Repeating the calculation
with two different integration tolerances produces relative variations
below \(9\times10^{-9}\) in the evolved variables.

The numerical solutions approach the asymptotic form
\begin{align}
 m(r)&=
 m_0+
 \frac{b^2r^3}{3}
 \left[
 \ln\left(\mathcal{K}r^{-3/2}\right)
 +\frac12
 \right]
 +o\!\left(r^3\ln r\right),
 \label{eq:interior_asymptotic_m}\\
 w(r)&=
 w_0-2\gamma\sqrt{r}
 +\mathcal{O}(r),
 \label{eq:interior_asymptotic_w}\\
 \mathcal{H}(r)&=
 \mathcal{H}_0+\mathcal{O}(\sqrt{r}),
 \label{eq:interior_asymptotic_H}\\
 \Xi(r)&=
 \mathcal{K}r^{-3/2}
 \left[1+\mathcal{O}(\sqrt{r})\right],
 \label{eq:interior_asymptotic_Xi}\\
 S(r)&=
 S_0\left[
 1+\frac{4\gamma^2}{\mathcal{K}}\sqrt{r}
 +\mathcal{O}(r)
 \right].
 \label{eq:interior_asymptotic_S}
\end{align}
The leading coefficients satisfy
\begin{equation}
 \gamma=
 \frac{1-w_0^2}{2\sqrt{m_0}},
 \qquad
 \mathcal{K}=
 \frac{2m_0\gamma}{\mathcal{H}_0}.
 \label{eq:interior_coefficient_relations}
\end{equation}
At \(r/r_h=10^{-8}\), the coefficients extracted from the numerical
solutions agree with Eq.~\eqref{eq:interior_coefficient_relations} to
approximately \(3\times10^{-4}\).

Representative limiting values are listed in
Table~\ref{tab:interior_parameters}.

\begin{table}[htbp]
\centering
\caption{Limiting interior parameters of the representative one-node
solutions with \(r_h=1\).}
\label{tab:interior_parameters}
\begin{tabular}{ccc}
\hline\hline
\(b\) & \(m_0\) & \(w_0\)\\
\hline
0.25  & 0.459630 & 0.973648\\
0.50  & 0.448791 & 0.983250\\
1     & 0.458235 & 0.987379\\
2     & 0.473218 & 0.990431\\
100   & 0.478129 & 0.993458\\
\hline\hline
\end{tabular}
\end{table}

The corresponding metric function behaves as
\begin{equation}
 N(r)=
 -\frac{2m_0}{r}
 +1+\mathcal{O}\!\left(r^2\ln r\right).
 \label{eq:interior_asymptotic_N}
\end{equation}

This finiteness of $m_0$, despite $X\to\infty$, follows directly from Eq.~(126). Since
$\Xi=1+X/b^2\sim Kr^{-3/2}$, the mass equation gives
\begin{equation}
\begin{aligned}
m'&=b^2r^2\ln\Xi\
&\sim b^2r^2\left(\ln K-\frac{3}{2}\ln r\right)
\longrightarrow0,
\qquad r\to0.
\end{aligned}
\label{eq:mass_integrability}
\end{equation}
Thus, the logarithmic growth of the effective energy density remains integrable near the center, and $m(r)$ approaches the finite limit $m_0$. In the linear theory, by contrast, the mass source grows nonintegrably in the corresponding strong-field regime, providing the mechanism underlying the oscillatory mass-inflation behavior of ordinary EYM black holes. The logarithmic nonlinearity therefore suppresses mass inflation without rendering the center regular. Since $m_0\neq0$, a Schwarzschild-type curvature singularity survives, as confirmed explicitly by the curvature invariant below.

Since \(m_0>0\) for all the solutions considered, \(N(r)\) remains
strictly negative sufficiently close to the center. The numerical
integration establishes \(N<0\) from \(r/r_h=10^{-8}\) up to the
immediate interior of the event horizon, while the asymptotic
expansion covers the remaining interval toward \(r=0\). The
representative solutions therefore possess a single nonextremal event
horizon and no inner Cauchy horizon.
\begin{figure}[H]
\centering
\includegraphics[width=\linewidth]
{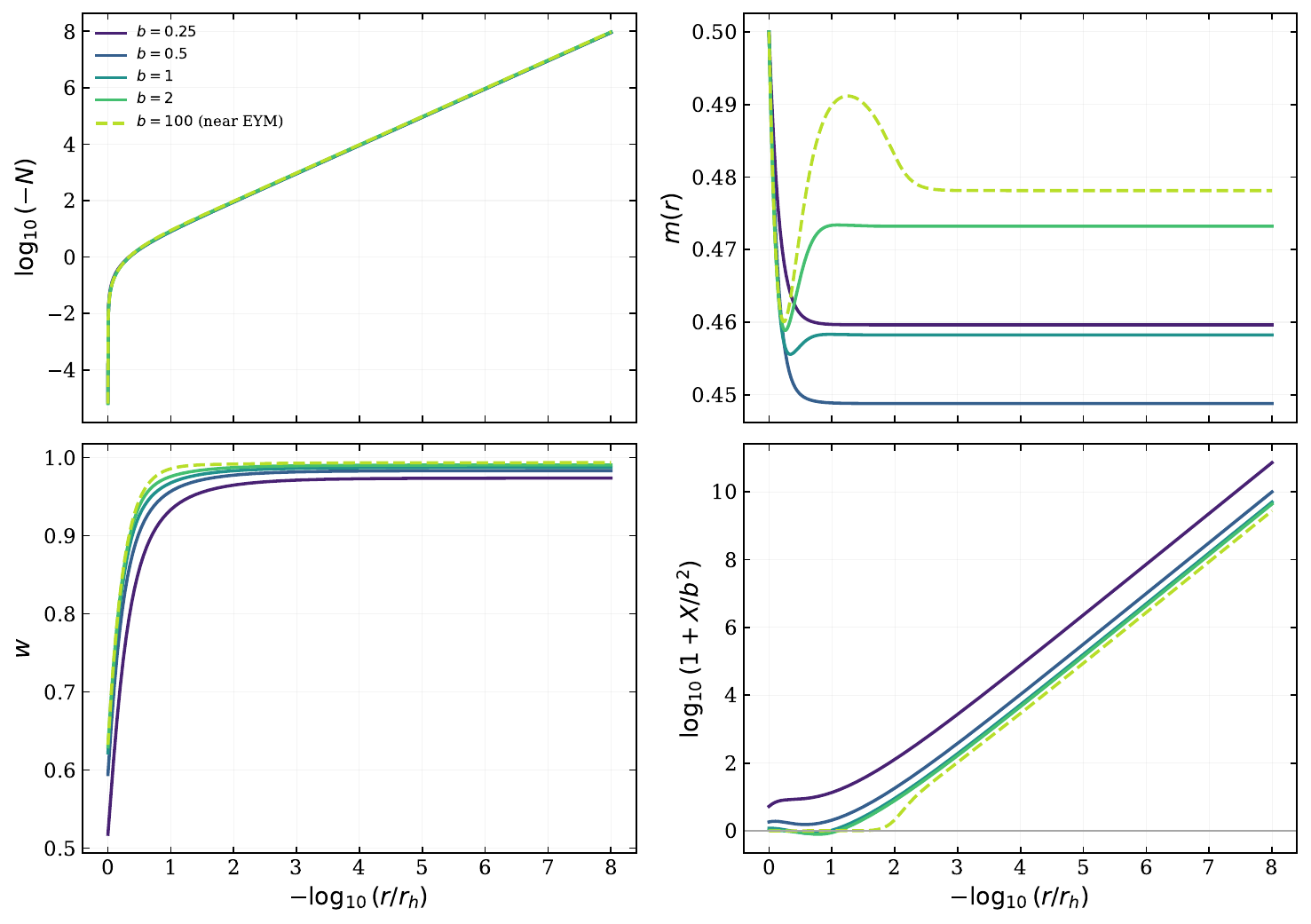}
\caption{Interior behavior of the fundamental one-node colored
black-hole solutions with \(r_h=1\), for
\(b=0.25,\,0.5,\,1,\,2,\) and \(100\). The horizontal coordinate is
\(-\log_{10}(r/r_h)\), extending from immediately inside the event
horizon to \(r/r_h=10^{-8}\). The upper-left panel shows that
\(\log_{10}(-N)\) approaches a straight line with unit slope,
corresponding to \(N\sim-2m_0/r\). The mass function approaches a
positive constant \(m_0\), while the gauge amplitude tends to a finite
value \(w_0\). The asymptotically linear behavior of
\(\log_{10}(1+X/b^2)\), with slope \(3/2\), confirms
\(1+X/b^2\sim\mathcal{K}r^{-3/2}\). No additional zero of \(N(r)\) is
found for the representative solutions considered.}
\label{fig:interior}
\end{figure}
\FloatBarrier
The limiting behavior also establishes that these solutions are not globally regular black holes. In dimensionless variables, the Kretschmann scalar has the leading behavior
\begin{equation}
 \ell_{\rm EYM}^{\,4}
 R_{\mu\nu\rho\sigma}R^{\mu\nu\rho\sigma}
 \sim
 \frac{48m_0^2}{r^6},
 \qquad r\rightarrow0,
 \label{eq:interior_Kretschmann}
\end{equation}
while the dimensionless Yang--Mills invariant diverges as
\begin{equation}
 X\sim b^2\mathcal{K}r^{-3/2}.
 \label{eq:interior_X}
\end{equation}
A regular center would instead require
\(m(r)=\mathcal{O}(r^3)\) and \(N(0)=1\). Since \(N<0\) throughout the interior, such a regular core is excluded. The interior therefore terminates at a Schwarzschild-type spacelike curvature singularity.

Although the \(b=100\) solution is already close to EYM in the
exterior, the convergence to the EYM solution is not uniform in the
radial coordinate. For every finite \(b\), the ratio \(X/b^2\)
eventually becomes large as \(r\rightarrow0\), and the solution enters
the genuinely logarithmic regime. Taking \(b\rightarrow\infty\) at
fixed radius therefore yields the strictly linear EYM equations,
whereas taking \(r\rightarrow0\) at fixed finite \(b\) yields the
asymptotic behavior in
Eqs.~\eqref{eq:interior_asymptotic_m}--\eqref{eq:interior_asymptotic_N}.
In this precise sense, the linear and central limits do not commute.
This nonuniformity explains why the deep interior at finite \(b\)
need not reproduce the oscillatory mass-inflation behavior of the
strictly linear EYM theory.

This behavior has a close qualitative parallel in non-Abelian
Einstein--Born--Infeld theory. Generic colored black holes in the
ordinary EYM model exhibit violent interior oscillations of the mass
function, whereas the Born--Infeld deformation was shown to suppress
this oscillatory regime~\cite{Dyadichev:2000dy}. Our results indicate
that the logarithmic deformation produces an analogous qualitative
effect: for every finite value of \(b\) examined, the interior
evolution approaches the nonoscillatory singular asymptotics derived
above. This suggests that the suppression of the characteristic EYM
interior oscillations is not peculiar to the Born--Infeld square-root
action, but may be a more general consequence of sufficiently strong
nonlinear gauge dynamics.

\section{Thermodynamic behavior and local stability}
\label{sec:thermo}

For the metric \eqref{eq:metric}, the Hawking temperature measured with
respect to the asymptotically normalized time coordinate is
\begin{equation}
 T_H=\frac{S_hN'_h}{4\pi\ell_{\rm EYM}}.
 \label{eq:temperaturephysical}
\end{equation}
In the numerical integration, the unnormalized horizon value is set to
\(S_h=1\), and the redshift function is subsequently divided by its
asymptotic value \(S_\infty\). The dimensionless temperature used below
is therefore
\begin{equation}
 \tau\equiv\ell_{\rm EYM}T_H
 =\frac{N'_h}{4\pi S_\infty}.
 \label{eq:temperaturedimensionless}
\end{equation}

Since the gravitational sector is described by the Einstein--Hilbert
action and the nonlinear Yang--Mills field is minimally coupled, the
black-hole entropy continues to satisfy the Bekenstein--Hawking area
law,
\begin{equation}
 S_{\rm BH}
 =
 \frac{\pi R_h^2}{G}
 =
 \frac{\pi\ell_{\rm EYM}^2r_h^2}{G}
 =
 \frac{4\pi^2}{g^2}r_h^2,
 \qquad
 R_h=\ell_{\rm EYM}r_h.
 \label{eq:entropyphysical}
\end{equation}
It is convenient to define the dimensionless entropy
\begin{equation}
 \mathcal{S}
 \equiv
 \frac{g^2}{4\pi}S_{\rm BH}
 =
 \pi r_h^2.
 \label{eq:entropydimensionless}
\end{equation}
The first law along a family at fixed \(b\) then takes the dimensionless
form
\begin{equation}
 \dd M=\tau\,\dd\mathcal{S}
 =2\pi r_h\tau\,\dd r_h.
 \label{eq:firstlawdimensionless}
\end{equation}
We verified Eq.~\eqref{eq:firstlawdimensionless} numerically, obtaining
a maximum relative discrepancy smaller than
\begin{equation}
 6.2\times10^{-4}
\end{equation}
over the families considered below.

The dimensionless heat capacity at fixed logarithmic parameter is
defined by
\begin{equation}
 C_b
 \equiv
 \left(\frac{\partial M}{\partial\tau}\right)_b
 =
 \frac{\dd M/\dd r_h}{\dd\tau/\dd r_h}
 =
 \frac{2\pi r_h\tau}{\dd\tau/\dd r_h}.
 \label{eq:heatcapacity}
\end{equation}
Its relation to the physical heat capacity is
\begin{equation}
 C_b^{\rm phys}
 =
 \left(\frac{\partial M_{\rm phys}}{\partial T_H}\right)_b
 =
 \frac{4\pi}{g^2}C_b.
 \label{eq:heatcapacityphysical}
\end{equation}
Since the prefactor is positive, the dimensionless and physical heat
capacities have the same sign and diverge at the same points.

We also introduce the dimensionless on-shell Helmholtz free energy
\begin{equation}
 \mathcal{F}
 \equiv
 \frac{G}{\ell_{\rm EYM}}F_{\rm phys}
 =
 M-\tau\mathcal{S}
 =
 M-\pi r_h^2\tau.
 \label{eq:freeenergy}
\end{equation}
The quantities \(C_b\) and \(\mathcal{F}\) were evaluated along the
one-node families by interpolating \(M(r_h)\) and \(\tau(r_h)\) at
fixed \(b\). The resulting thermodynamic behavior is displayed in
Fig.~\ref{fig:transitions}.

\begin{figure}[!htbp]
\centering
\includegraphics[width=\linewidth]
{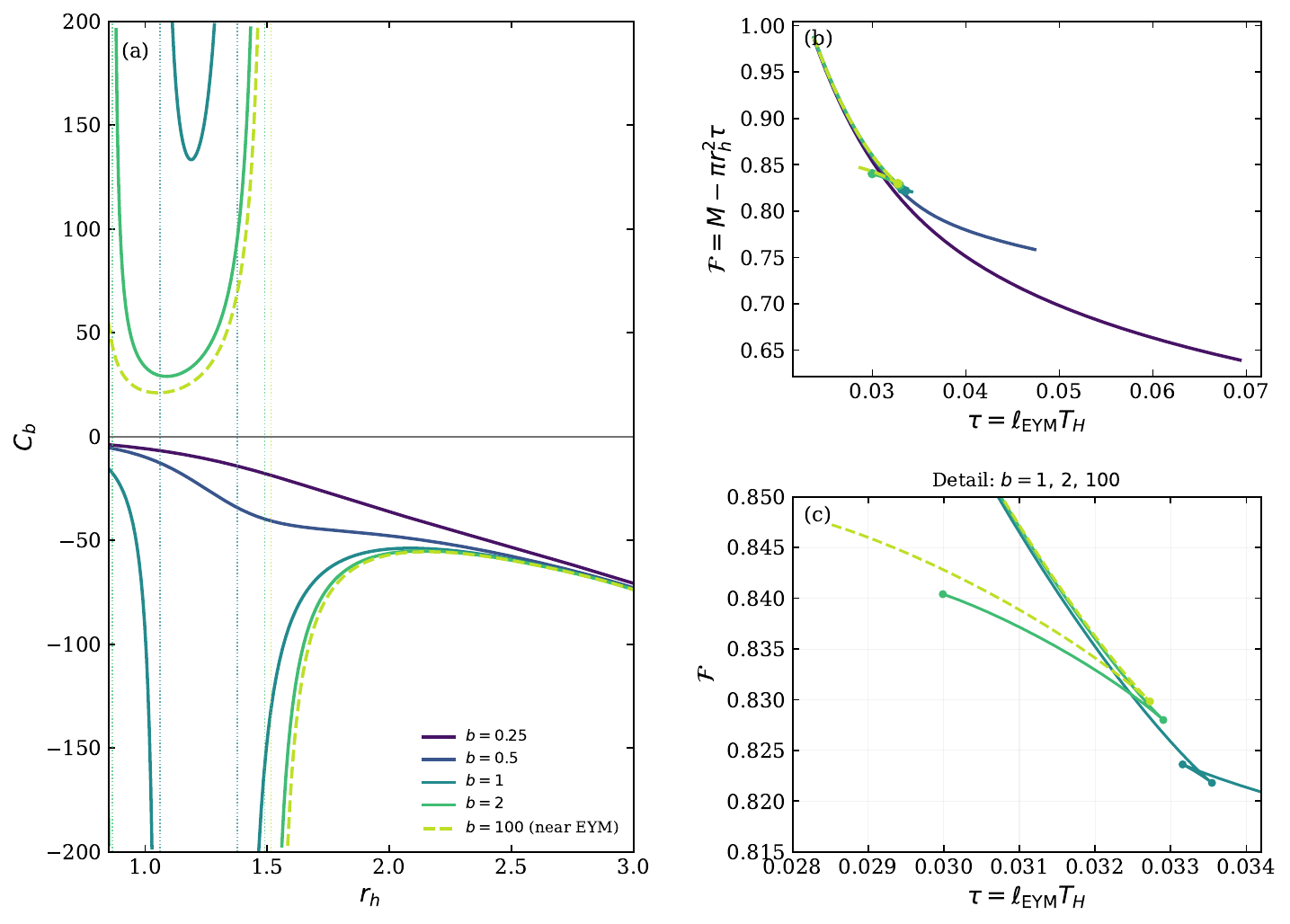}
\caption{Thermodynamic behavior of the one-node (\(n=1\)) colored
black-hole solutions for \(b=0.25,\,0.5,\,1,\,2,\) and \(100\), with
\(0.85\leq r_h\leq3\) and outer integration radius \(R=80\).
Panel (a) shows the dimensionless heat capacity
\(C_b=\dd M/\dd\tau\); the vertical dotted lines mark the temperature
turning points \(\dd\tau/\dd r_h=0\), where \(C_b\) diverges and its
sign changes. Panel (b) shows the dimensionless Helmholtz free energy
\(\mathcal{F}=M-\pi r_h^2\tau\) on the full scale, while panel (c)
enlarges the region containing the turning points for
\(b=1,\,2,\) and \(100\). The case \(b=100\) represents the near-EYM
limit.}
\label{fig:transitions}
\end{figure}
For \(b=0.25\) and \(0.5\), the temperature decreases throughout the
investigated interval, so that the corresponding heat capacity is
negative. Within the local canonical criterion, these branches are
therefore thermodynamically unstable. For \(b=1\), the temperature
develops a local minimum followed by a local maximum. The intermediate
branch between these turning points has \(C_b>0\), whereas the branches
on either side have \(C_b<0\).

A similar three-branch structure occurs for \(b=2\), with two
temperature turning points separating intervals of alternating signs
of the heat capacity. The lower turning point persists after extending
the numerical family toward smaller horizons and refining the
\(r_h\)-grid. In the near-EYM case \(b=100\), the temperature increases
up to a single maximum and then decreases, producing a
positive-heat-capacity branch at smaller \(r_h\) and a
negative-heat-capacity branch at larger \(r_h\). 

This qualitative pattern, namely, a positive-specific-heat branch bounded by two
temperature turning points, flanked by negative-specific-heat branches on
either side, is already known for ordinary EYM colored black holes and
related Einstein--Yang--Mills--dilaton solutions~\cite{Torii:1993vm}. Our
near-EYM case $b=100$ therefore reproduces the expected thermodynamic shape
of the $b\to\infty$ limit, providing a qualitative cross-check of the
thermodynamic branch structure that complements the quantitative benchmark
of Sec.~\ref{sec:numerics}.

The numerical locations of the temperature extrema are summarized in
Table~\ref{tab:critical_points}.

\begin{table}[htbp]
\centering
\caption{Temperature turning points of the fundamental one-node
colored black-hole families. The critical radii satisfy
\(d\tau/dr_h=0\), where the canonical heat capacity diverges.
The quoted precision reflects variations under refinement of the
horizon-radius grid and changes of the outer integration boundary.}
\label{tab:critical_points}
\begin{tabular}{ccc}
\hline\hline
\(b\) & \(r_h^{\rm crit}\) & \(\tau_{\rm crit}\) \\
\hline
1   & 1.0608 & 0.033159 \\
1   & 1.3776 & 0.033546 \\
2   & 0.8655 & 0.029987 \\
2   & 1.4901 & 0.032903 \\
100 & 1.5162 & 0.032727 \\
\hline\hline
\end{tabular}
\end{table}

The robustness of these turning points was examined by refining the
horizon-radius grid and by repeating the calculation with different
outer boundaries. For \(b=1\), reducing the grid spacing from
\(\Delta r_h=0.025\) to \(0.0125\) changes the extracted critical
radii by less than \(3\times10^{-6}\). Increasing the outer boundary
from \(R=80\) to \(R=120\) produces shifts below \(6\times10^{-5}\).
A comparable local refinement confirms the lower turning point for
\(b=2\): direct neighboring values of \(\tau(r_h)\) decrease toward
\(r_h\simeq0.8655\) and increase again below it. Hence this minimum is
not generated solely by spline interpolation. The number of digits
reported in Table~\ref{tab:critical_points} has been restricted
accordingly.

At each critical radius, the heat capacity diverges and changes sign,
separating branches with different local thermodynamic stability.
These points can therefore be interpreted as Davies-type local phase
transitions, or, more precisely, as transitions of local
thermodynamic stability.

This local interpretation should not be confused with a global phase
transition in a canonical ensemble. Asymptotically flat black holes
cannot, in general, remain in stable thermal equilibrium with a heat
bath at infinity, and the on-shell free energy
\(\mathcal{F}\) does not by itself define a globally well-posed
canonical ensemble. Moreover, Fig.~\ref{fig:transitions} does not
display an unambiguous free-energy crossing between two globally
stable competing phases. A genuine global phase-transition analysis
would require an additional thermodynamic prescription, such as
placing the black hole inside a finite cavity or extending the model
to an asymptotically AdS background.

\section{Conclusions}
\label{sec:conclusions}

We have constructed static, spherically symmetric, and asymptotically
flat colored black holes in four-dimensional Einstein gravity coupled
to logarithmic nonlinear \(SU(2)\) Yang--Mills theory. In contrast with
the analytical black holes obtained from the Wu--Yang magnetic ansatz,
the present configurations contain a dynamical radial amplitude
\(w(r)\) and arise from the full coupled nonlinear boundary-value
problem. Their gauge field approaches a vacuum value,
\(w(\infty)=\pm1\), and therefore carries no independently specifiable
asymptotic Yang--Mills magnetic charge. The resulting non-Abelian hair
is instead characterized by the gauge profile and its discrete node
number.

The local analysis at the event horizon provided the regular initial
data required for the numerical shooting construction. At infinity,
the solutions approach the Schwarzschild geometry at leading order,
but retain a definite gravitational imprint of the colored field. The
asymptotic gauge amplitude decays as \(1/r\), while the leading
hair-induced correction to the metric appears at order \(r^{-4}\).
The explicitly logarithmic contribution is much more strongly
suppressed, entering the metric only at order
\(\mathcal{O}(b^{-2}r^{-10})\). The coefficient of the gauge-field tail
is not an independent conserved charge, but secondary hair fixed by
the horizon data and the nonlinear shooting condition.

An integral identity establishes that a nontrivial solution cannot
remain sign-definite while approaching a Yang--Mills vacuum in the
same sign sector. Consequently, colored configurations connecting the
two vacuum sectors must possess at least one zero of \(w(r)\). We
constructed the fundamental one-node branch and verified that it
approaches continuously the ordinary EYM colored black hole as
\(b\rightarrow\infty\). A direct calculation in the strictly linear
theory reproduces the published EYM horizon value and agrees with the
\(b=100\) solution at the level of approximately \(2.1\times10^{-5}\)
or better for the quantities considered, providing an independent
validation of the numerical implementation.

The logarithmic interaction produces substantial modifications in the
strong-field regime. At fixed \(r_h=1\), decreasing \(b\) lowers the
ADM mass, raises the Hawking temperature, and displaces the
Yang--Mills node away from the horizon. The fundamental branch was
continued down to \(b=2\times10^{-3}\), with no evidence of a finite
lower critical value within the explored interval. As \(b\) decreases,
the exterior geometry approaches Schwarzschild on any fixed radial
domain, whereas the node and the colored-field tail are pushed toward
progressively larger radii. This behavior points to a nonuniform
decoupling limit in which the gravitational geometry becomes locally
Schwarzschild-like while the non-Abelian structure survives at
increasingly distant scales. Establishing the strict \(b\rightarrow0\)
limit would, however, require numerical control over an unbounded
radial domain.

Some of the qualitative conclusions extend beyond the particular
logarithmic Lagrangian considered here. The nodal restriction derived
in Sec.~\ref{sec:nodal} relies only on the positivity of \(P=-\mathcal{L}_X\), whereas
the suppression of interior mass inflation is associated with the
sufficiently slow growth of \(\mathcal{L}(X)\) at large field strength.
These results suggest that both features may persist within a broader
class of nonlinear Yang--Mills theories, although their precise range
of validity must be established separately for each nonlinear
deformation.

The continuation through the event horizon reveals a qualitatively
different interior from that of the strictly linear EYM black hole.
For all representative finite values of \(b\), no additional zero of
\(N(r)\) was found, so the solutions possess a single nonextremal event
horizon and no inner Cauchy horizon. Near the center, the mass function
approaches a positive finite value, the gauge amplitude remains
finite, and the logarithmic invariant diverges according to
\(1+X/b^2\sim Kr^{-3/2}\). The center is therefore not regular:
the Kretschmann scalar behaves as \(48m_0^2/r^6\), corresponding to a
Schwarzschild-type spacelike curvature singularity. The logarithmic
nonlinearity suppresses the oscillatory mass-inflation behavior
characteristic of ordinary EYM interiors, without resolving the
central singularity. Moreover, the limits \(b\rightarrow\infty\) and
\(r\rightarrow0\) do not commute, since every finite-\(b\) solution
eventually enters the genuinely logarithmic regime sufficiently close
to the center.

The thermodynamic families display turning points of the Hawking
temperature and corresponding divergences of the heat capacity. For
\(b=1\) and \(b=2\), two turning points separate three branches with
alternating signs of \(C_b\), whereas the near-EYM case \(b=100\)
contains a single temperature maximum. These turning points are robust
under refinement of the horizon-radius grid and variation of the
outer numerical boundary. They may be interpreted as Davies-type
transitions of local thermodynamic stability. They do not, however,
establish a global phase transition: an asymptotically flat black hole
does not define a globally stable canonical ensemble without an
additional prescription, and the calculated Helmholtz free energy
shows no unambiguous crossing between globally stable competing
phases.

The present results establish the existence and principal geometric
properties of the fundamental colored branch, but do not determine
its dynamical stability. Since the solutions approach the known
unstable asymptotically flat EYM colored black holes as
\(b\rightarrow\infty\), the corresponding gravitational and
sphaleronic unstable modes are expected to persist at least in the
near-EYM regime. Whether sufficiently strong logarithmic
nonlinearities modify the number or growth rates of these modes
requires a dedicated linear perturbation analysis. Further natural
extensions include constructing the higher-node branches, mapping
their complete domains of existence, and investigating the theory in
a finite cavity or an asymptotically AdS background, where a
well-defined global thermodynamic phase structure can be formulated.

\section*{Acknowledgements}
\noindent The author would like to thank the Conselho Nacional de Desenvolvimento Cient\'{i}fico e Tecnol\'{o}gico (CNPq) for partial financial support, through grant 301122/2025-3.

\bibliographystyle{unsrt}
\bibliography{Refer}

\end{document}